\documentclass{emsprocart}

\contact[o.turek@ujf.cas.cz]{Ond\v{r}ej Turek, Nuclear Physics Institute, Academy of Sciences of the Czech Republic, Hlavn\'{i} 130, \v{R}e\v{z}, Czech Republic, \emph{and} Joint Institute for Nuclear Research, Dubna, Russia}

\usepackage[pdftex]{graphicx}

\newtheorem{theorem}{Theorem}[section]

\newtheorem{lemma}[theorem]{Lemma}

\theoremstyle{definition}

\newtheorem{remark}[theorem]{Remark}

\def\RR{\mathbb R}
\def\CC{\mathbb C}

\def\im{\mathrm i}
\def\exp{\mathrm e}
\def\Sc{\mathcal{S}}
\def\Trans{\mathcal{T}}
\def\Prob{\mathcal{P}}

\def\diss{\mathrm{diss}}
\def\adj{\mathrm{adj}}
\def\const{\mathrm{const}}
\newcommand{\diag}{\operatorname{diag}}
\newcommand{\rank}{\operatorname{rank}}
\newcommand{\Tr}{\operatorname{Tr}}
\newcommand{\sgn}{\operatorname{sgn}}

\title[On quantum graph filters with flat passbands]{On quantum graph filters with flat passbands}

\author[Ond\v{r}ej Turek]{Ond\v{r}ej Turek\thanks{The author is grateful to
Prof.~Taksu Cheon for inspirative discussions on the topic, and to Kochi University of Technology for hospitality during the writing of
this paper. The work was supported by Youth JINR Grant No.\ 15-302-08.}}

\begin{document}

\begin{abstract}
We examine transmission through a quantum graph vertex to which auxiliary edges with constant potentials are attached. We find a characterization of vertex couplings for which the transmission probability from a given ``input'' line to a given ``output'' line shows a flat passband. The bandwidth is controlled directly by the potential on the auxiliary edges. Vertices with such couplings can thus serve as controllable band-pass filters. The paper extends earlier works on the topic. The result also demonstrates the effectivity of the $ST$-form of boundary conditions for a study of scattering in quantum graph vertices.
\end{abstract}



\maketitle

\section{Introduction}

Quantum mechanics on graphs is a useful tool for the examination of quantum motion on microscopic wires, lattices and other graph-like nanostructures. The method has been intensively developed since 1980s with regard to the technological progress achieved in microfabrication. The development of the subject led to rich literature to date; see e.g.\ monographs~\cite{AG05,EKST08} and references therein. On the other hand, the discipline remains relatively new and is still rapidly advancing.

Quantum graph models are useful in particular for a design of quantum systems with prescribed properties. In this paper we focus on scattering problems on systems consisting of several wires connected together in one point to form a star. When a particle moving along a wire reaches the vertex, it is scattered to the other wires. The scattering characteristics depend on the energy of the particle and on the nature of the potential in the point. Such system is modelled by a star graph with a certain wave function coupling in the vertex. It is known that a vertex of degree $n$ generally features $n^2$-parametric family of admissible couplings~\cite{KS99}, and the scattering characteristics considerably vary in dependence on the coupling parameters~\cite{CET09}. Obviously, one can take advantage of this fact in a design of quantum devices with particular particle transmission properties. On the other hand, the role of the coupling parameters in the scattering characteristics is not well understood yet.

The problem studied in this paper concerns a star graph with $n$ edges, some of which being subject to a constant nonzero potential $V$. Scattering in such a system depends i.a.\ on the strength of the potential. It was noticed in earlier works~\cite{TC12,TC13} that a certain particular choice of the vertex coupling gives rise to a ``flat band'' scattering behaviour. That is, the probability of transmission of a particle from an edge (we call the edge ``input'') to another given edge (called ``output'') turned out to be constant for energies $E$ in the interval $(0,V)$ and quickly descending towards zero for $E>V$. Consequently, particles with energies exceeding the controlling potential $V$ mostly cannot pass to the output edge. The vertex thus works as a controllable band-pass filter with a flat passband. In this paper we will deal with the problem more thoroughly. We will prove that this behaviour can occur only for certain subfamilies of vertex couplings, but at the same time we will demonstrate that there exists a multiparametric family of vertex couplings with the ``flat passband'' property. In other words, such a behaviour is less rare than it might seem earlier. 

The paper is organized as follows. In Section~\ref{Section: Preliminaries} we bring together elementary facts and notation on vertex couplings and scattering in quantum graph vertices. Section~\ref{Sect.Filter} presents the concept of a controllable band-pass filter and the goals of the paper, as well as the idea of solution. Sections~\ref{Sect.1}--\ref{Sect.3} are devoted to the existence of quantum graph filters featuring flat passbands. The main result is presented in Section~\ref{Sect.2}, in which several designs are proposed.

\section{Preliminaries}\label{Section: Preliminaries}

A wave function of a particle confined to a star graph having $n$ arms consists of $n$ components, $\Psi=(\psi_1,\psi_2,\ldots,\psi_n)$. The coordinate on each arm is chosen such that $0$ corresponds to the center of the star graph and the variable grows in the outgoing direction. If there are potentials $V_1,\ldots,V_n$ imposed on the arms, the Hamiltonian acts as $\psi_j\mapsto-\psi_j''+V_j\psi_j$ at each arm $j=1,\ldots,n$ (we choose the units so that $\hbar=2m=1$ for $m$ being the mass of the particle).

Properties of the vertex are determined by boundary conditions that are conventionally written in the form
\begin{equation}\label{b.c.}
A\Psi(0)+B\Psi'(0)=0\,,
\end{equation}
where
\begin{equation}
\Psi(0)=
\begin{pmatrix}
\psi_1(0) \\
\vdots \\
\psi_n(0)
\end{pmatrix}
\qquad\text{and}\qquad
\Psi'(0)=
\begin{pmatrix}
\psi'_1(0) \\
\vdots \\
\psi'_n(0)
\end{pmatrix}
\end{equation}
are the boundary vectors and $A,B$ are complex $n\times n$ matrices satisfying
\begin{equation}\label{KS}
\rank(A|B)=n\,, \qquad AB^*=BA^*\,,
\end{equation}
cf.~\cite{KS99}. The symbol $(A|B)$ denotes the $n\times2n$ matrix formed from columns of $A$ and $B$.

In this paper we will take advantage of the so-called $ST$-form of boundary conditions~\cite{CET10}, in which requirements~\eqref{KS} are implicitly satisfied due to a special choice of $A$ and $B$. Namely, the $ST$-form relies on the block decomposition of $A$ and $B$,
\begin{equation}\label{ST}
\begin{pmatrix}
I^{(r)} & T \\
0 & 0
\end{pmatrix}
\Psi'(0)=
\begin{pmatrix}
S & 0 \\
-T^* & I^{(n-r)}
\end{pmatrix}
\Psi(0)
\end{equation}
for a certain $r\in\{0,1,\ldots,n\}$. Matrix $T$ is a general complex $r\times n-r$ matrix, $S$ is a Hermitian matrix of order $r$, and $I^{(r)},I^{(n-r)}$ are identity matrices of given orders. The value $r$ corresponds to $\rank(B)$ in boundary conditions \eqref{b.c.}.

If a wave corresponding to a quantum particle with energy $E$ reaches the vertex from the $\ell$-th line with amplitude $1$, the wave is reflected with a complex amplitude $\Sc_{jj}(E)$ and transmitted to the lines nos.~$\ell\neq j$ with complex amplitudes $\Sc_{j\ell}(E)$. The scattering amplitudes form the scattering matrix of the vertex. The scattering matrix, denoted by $\Sc(E)$, is an $n\times n$ matrix function of particle energy that is given by the formula
\begin{equation}\label{S}
\Sc(E) = - (A + \im\sqrt{E} B)^{-1}(A - \im\sqrt{E} B)\,.
\end{equation}
Let us emphasize that formula~\eqref{S} applies only if $V_j=0$ for all $j=1,2,\ldots,n$.
If we substitute
$$
A=-\begin{pmatrix}
S & 0 \\
-T^* & I^{(n-r)}
\end{pmatrix}
\quad\text{and}\quad
B=\begin{pmatrix}
I^{(r)} & T \\
0 & 0
\end{pmatrix}
$$
into equation~\eqref{ST}, we obtain the scattering matrix expressed in terms of the $ST$-form of boundary conditions,
\begin{equation}\label{S ST}
\Sc(E)=-I^{(n)}+
2\begin{pmatrix} I^{(r)} \\ T^* \end{pmatrix}
\left(I^{(r)}+TT^*-\frac{1}{\im\sqrt{E}}S\right)^{-1}
\begin{pmatrix} I^{(r)} & T\end{pmatrix}\,;
\end{equation}
cf.~\cite{CET10b}. It is straightforward to see from formula~\eqref{S ST} that the scattering matrix is constant with respect to $E$ if and only if the matrix $S$ in the $ST$-form of boundary conditions~\eqref{ST} vanishes, i.e., when 
\begin{equation}\label{FT}
\begin{pmatrix}
I^{(r)} & T \\
0 & 0
\end{pmatrix}
\Psi'(0)=
\begin{pmatrix}
0 & 0 \\
-T^* & I^{(n-r)}
\end{pmatrix}
\Psi(0)\,.
\end{equation}
Vertex couplings having energy-independent scattering matrices are called scale invariant couplings. They are widely studied; see~\cite{FT00,NS00,SS02,CT10}.

\section{A potential-controlled filter}\label{Sect.Filter}

Consider a quantum star graph with $n$ edges. We will regard one of the edges as \emph{input}, another edge as \emph{output}. The remaining $n-2$ edges will be assumed to be of two types, see Figure~\ref{Filter}:

\begin{itemize}
\item Lines with constant nonzero potentials (\emph{``controlling lines''});
\item lines without potentials (\emph{``drains''}).
\end{itemize}
\begin{figure}[h]
\begin{center}
\includegraphics[angle=0,width=0.4\textwidth]{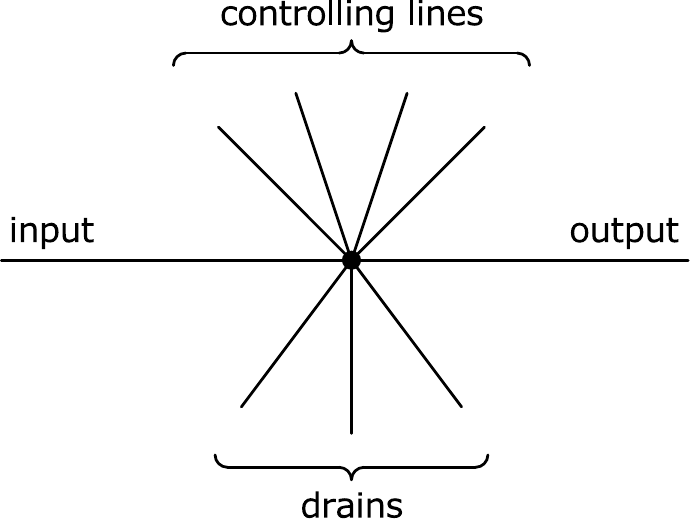}
\caption{A controllable quantum graph filter}\label{Filter}
\end{center}
\end{figure}
For a particle coming in the vertex along the input line with energy $E$, we denote the complex transmission amplitude to the output line by the symbol $\Trans(E)$. The corresponding transmission probability in the channel is $\Prob(E)=|\Trans(E)|^2$. This paper is concerned with the relation between the transmission probability in the input-output channel and the potentials on the controlling lines. More specifically,
we will search for couplings that can serve as controllable band-pass filters with flat passbands. That is, the function $\Prob(E)$ is required to have the following three properties, cf.~Figure~\ref{Fig.passband}.
\begin{gather}
\Prob(E)=\const>0 \quad \text{for } E\in(0,V) \text{ for a certain } V>0; \label{const} \\
\Prob(E) \text{ quickly decreases when $E$ exceeds $V$, i.e., } \lim_{E\searrow V}\Prob'(E)=-\infty; \label{drop} \\
\lim_{E\to\infty}\Prob(E)=0. \label{lim=0}
\end{gather}
We assume that $V$ is the value of the constant potential on the controlling lines.
\begin{figure}[h]
\begin{center}
\includegraphics[angle=0,width=0.45\textwidth]{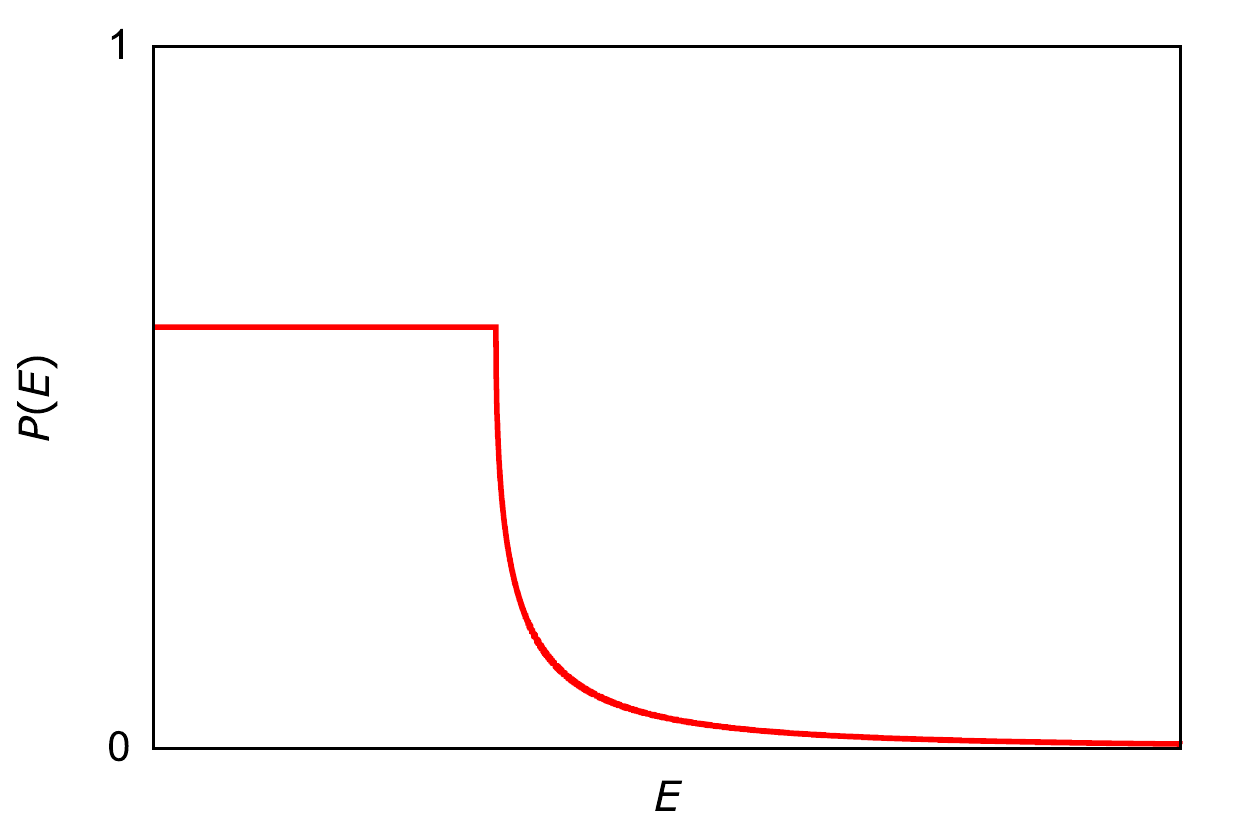}
\caption{An example of sought transmission probability}\label{Fig.passband}
\end{center}
\end{figure}

In general, the transmission amplitude in the input-output channel is given by the term $[\Sc(E)]_{oi}$ of the scattering matrix. However, formula~\eqref{S} cannot be applied straightforwardly, because the controlling lines are subject to constant potentials $V_j\neq0$. Therefore, we will approach the problem as follows. At first we transform the original boundary conditions in the vertex of degree $n$ to boundary conditions in a vertex of degree $2$. This step is based on the idea that the controlling lines and drains support only outgoing waves, thus the corresponding wave function components are multiples of $\exp^{\im k_j x}$, where
\begin{equation}\label{k_j}
k_j=\begin{cases}
\sqrt{E-V_j} & \text{if } E>V_j\,; \\
\im\sqrt{V_j-E} & \text{if } E<V_j
\end{cases}
\end{equation}
is the momentum on the $j$-th line with potential $V_j$.
Relation $\psi_j(x)\propto\exp^{\im k_j x}$ implies
\begin{equation}\label{proportional}
\psi_j'(0)=\im k_j\psi_j(0) \quad\text{for all } j\neq i, j\neq o\,.
\end{equation}
Equation~\eqref{proportional} allows us to eliminate the boundary values $\psi_j(0)$ and $\psi_j'(0)$ at all controlling edges and drains from boundary conditions~\eqref{ST}. We obtain boundary conditions in a vertex of degree $2$ that connect just the input and the output,
\begin{equation}\label{dissipative}
A_\diss\Psi_{io}+B_\diss\Psi_{io}'=0\,,
\end{equation}
where
\begin{equation}\label{Psi_io}
\Psi_{io}=\begin{pmatrix}\psi_i(0) \\ \psi_o(0) \end{pmatrix}\,; \qquad \Psi'_{io}=\begin{pmatrix}\psi_i'(0) \\ \psi_o'(0) \end{pmatrix}\,,
\end{equation}
are boundary values at input and output line. We emphasize that $A_\diss,B_\diss$ are $2\times2$ matrices that generally do \emph{not} obey the requirements~\eqref{KS}, because $AB^*=BA^*$ can be broken due to the dissipation in the vertex, manifested through ``hidden'' drains and controllers. On the other hand, since neither input line nor the output line support a potential, formula~\eqref{S} applies without reserve. Once we subsitute matrices $A_\diss, B_\diss$ from reduced boundary conditions~\eqref{dissipative} into equation~\eqref{S}, we obtain the $2\times2$ scattering matrix that characterizes wave propagation in the input-output channel. In particular, the $(2,1)$-term of the matrix is the sought transmission amplitude $\Trans(E)$.

For the derivation of matrices $A_\diss,B_\diss$, we will take advantage of the $ST$-form of boundary conditions. Therefore, the calculation depends on the parameter $r$. In the following section we begin with the case $r=1$.

\section{Case $r=1$}\label{Sect.1}

The $ST$-form~\eqref{ST} of boundary conditions for $r=1$ uses matrices $T=(t_2 \; t_3 \; \cdots \; t_n)$ and $S=(s)$. We may assume without loss of generality that line no.~1 is the input and line no.~2 is the output. Let us follow the steps outlined in Section~\ref{Sect.Filter}. After eliminating $\psi_3(0),\ldots,\psi_n(0)$ from the system using identities~\eqref{proportional}, we get dissipative boundary conditions~\eqref{dissipative} with
$$
A_\diss=
-\begin{pmatrix}
s-\im\sum_{j=3}^n\frac{k_j}{\sqrt{E}}|t_j|^2 & 0 \\
-\overline{t_2} & 1
\end{pmatrix}
\,,\qquad
B_\diss=
\begin{pmatrix}
1 & t_2 \\
0 & 0
\end{pmatrix}\,.
$$
When we substitute $A_\diss,B_\diss$ into formula~\eqref{S}, we obtain the scattering matrix describing the input-output interface,
$$
\Sc_\diss(E)=-I+\frac{2}{1+|t_2|^2-\frac{s}{\im\sqrt{E}}+\sum_{j=3}^{n}\frac{k_j}{\sqrt{E}}|t_j|^2}\begin{pmatrix}
1 & t_2 \\
\overline{t_2} & |t_2|^2
\end{pmatrix}\,.
$$
The transmission amplitude is given as the term $[\Sc_\diss(E)]_{21}$, i.e.,
\begin{equation}\label{T1}
\Trans(E)=\frac{2\overline{t_2}}{1+|t_2|^2-\frac{s}{\im\sqrt{E}}+\sum_{j=3}^{n}\frac{k_j}{\sqrt{E}}|t_j|^2}\,.
\end{equation}

Now we are ready to check whether $S$ and $T$ can be chosen such that the function $\Prob(E)=|\Trans(E)|^2$ satisfies conditions~\eqref{const}--\eqref{lim=0}. Condition~\eqref{lim=0} is equivalent to $\lim_{E\to\infty}\Trans(E)=0$. Equation~\eqref{k_j} implies $\lim_{E\to\infty}\frac{k_j}{\sqrt{E}}=1$ for all $j=3,\ldots,n$; hence
$$
\lim_{E\to\infty}\Trans(E)=\frac{2\overline{t_2}}{1+|t_2|^2+\sum_{j=3}^{n}|t_j|^2}\,.
$$
Consequently
$$
\lim_{E\to\infty}\Trans(E)=0 \quad\Leftrightarrow\quad t_2=0\,.
$$
However, the choice $t_2=0$ implies $\Trans(E)=0$ for all $E>0$ (cf.~\eqref{T1}), which contradicts condition~\eqref{const}. (In physical terms, $t_2=0$ corresponds to a vertex with line no.~$2$ completely decoupled.) To sum up, conditions~\eqref{lim=0} and~\eqref{const} cannot be satisfied at the same time. We conclude that a band-pass filter with flat passband cannot be constructed using a vertex coupling with $r=1$.

\section{Case $r\geq2$ with linear dependence}\label{Sect.LD}

Now we consider boundary conditions~\eqref{b.c.} with $r=\rank(B)\geq2$ such that the columns of $B$ corresponding to the input and output are linearly dependent. We can assume without loss of generality that the input corresponds to line no.~1 and the output is line no.~$n$. When the boundary conditions are written in the $ST$-torm, the linear dependence implies that the last column of $T$ is a transposition of the vector $(t,0,\ldots,0)$ for a certain $t\neq0$. Therefore, the $ST$-form of boundary conditions reads as follows,
\begin{equation}\label{STdecLD}
\begin{pmatrix}
1 & 0 & T_1 & t \\
0 & I^{(r-1)} & T_2 & 0 \\
0 & 0 & 0 & 0 \\
0 & 0 & 0 & 0
\end{pmatrix}
\begin{pmatrix}
\psi_i' \\
\Psi'_{cd} \\
\psi_o'
\end{pmatrix}
=
\begin{pmatrix}
s & S_2 & 0 & 0 \\
S_2^* & S_4 & 0 & 0 \\
-T_1^* & -T_2^* & I^{(n-r-1)} & 0 \\
-\bar{t} & 0 & 0 & 1
\end{pmatrix}
\begin{pmatrix}
\psi_i \\
\Psi_{cd} \\
\psi_o
\end{pmatrix}
\end{equation}
where $I^{(r-1)},I^{(n-r-1)}$ are identity matrices of given orders and $T=\begin{pmatrix} T_1 & t \\
T_2 & 0\end{pmatrix}$, $S=\begin{pmatrix}s & S_2 \\ S_2^* & S_4\end{pmatrix}$. 
Recall that symbols $\psi_i,\psi_i'$ and $\psi_o,\psi_o'$ denote boundary values at the input and output line, respectively. Symbols
$$
\Psi_{cd}=\begin{pmatrix}\psi_{2}(0) \\ \vdots \\ \psi_{n-1}(0) \end{pmatrix}\,; \qquad \Psi'_{cd}=\begin{pmatrix}\psi_{2}'(0) \\ \vdots \\ \psi_{n-1}'(0) \end{pmatrix}
$$
stand for boundary vectors at controlling edges and drains.

Values $\psi_j(0),\psi_j'(0)$ obey relations~\eqref{proportional}, i.e.,
\begin{equation}\label{K_ld}
\Psi'_{cd}=\im\begin{pmatrix}K_2 & 0 \\ 0 & K_3\end{pmatrix}\Psi_{cd}
\end{equation}
for $K_2=\diag(k_2,\ldots,k_r)$ and $K_3=\diag(k_{r+1},\ldots,k_{n-1})$. We use identity~\eqref{K_ld} to eliminate $\Psi_{cd}$ and $\Psi'_{cd}$ from system~\eqref{STdecLD}. In this way we obtain boundary conditions~\eqref{dissipative} with
\begin{equation}\label{AB diss}
A_\diss=
-\begin{pmatrix}
f & 0 \\
-\bar{t} & 1
\end{pmatrix}
\,,\qquad
B_\diss=
\begin{pmatrix}
1 & t \\
0 & 0
\end{pmatrix}\,,
\end{equation}
where
$$
f=s-\im T_1K_3T_1^*+(S_2-\im T_1K_3T_2^*)(\im K_2+\im T_2K_3T_2^*-S_4)^{-1}(S_2^*-\im T_2K_3T_1^*)\,.
$$
The dissipative scattering matrix corresponding to matrices~\eqref{AB diss} is
$$
\Sc_\diss(E)=-I+\frac{2}{1+|t|^2-\frac{f}{\im\sqrt{E}}}
\begin{pmatrix}
1 & t \\
\overline{t} & |t|^2
\end{pmatrix}.
$$
The transmission amplitude thus equals
$$
\Trans(E)=\frac{2\bar{t}}{1+|t|^2-\frac{f}{\im\sqrt{E}}}\,.
$$
Now we check condition~\eqref{lim=0}. Since $\lim_{E\to\infty}\frac{1}{\sqrt{E}}K_2=I^{(r-1)}$ and $\lim_{E\to\infty}\frac{1}{\sqrt{E}}K_3=I^{(n-r-1)}$, we have
$$
\lim_{E\to\infty}\frac{-f}{\im\sqrt{E}}=T_1T_1^*-T_1T_2^*(I+T_2T_2^*)^{-1}T_2T_1^*=T_1(I+T_2^*T_2)^{-1}T_1^*\,.
$$
Hence
$$
\lim_{E\to\infty}\Trans(E)=\frac{2\overline{t}}{1+|t|^2+T_1(I+T_2^*T_2)^{-1}T_1^*}\,.
$$
Note that $T_1(I+T_2^*T_2)^{-1}T_1^*$ is a non-negative number for any choice of $T_1,T_2$. Therefore, condition~\eqref{lim=0} is equivalent to $t=0$.
However, it is easy to see from boundary conditions~\eqref{STdecLD} that $t=0$ corresponds to a completely decoupled output, which implies $\Trans(E)=0$ for all $E>0$. In other words, conditions~\eqref{lim=0} and~\eqref{const} are contradictory.
We conclude that a vertex coupling cannot serve for the construction of a band-pass filter with flat passband if the columns of $B$ corresponding to the input and output edge are linearly dependent.

\section{Case $r=2$}\label{Sect.2}

This section is focused on the case $r=2$, i.e., $\rank(B)=2$. With regard to the result of Section~\ref{Sect.LD}, we may assume that the columns of matrix $B$ that correspond to the input and output line are linearly independent. Without loss of generality, we associate the input and output with lines no.~$1$ and no.~$2$, respectively. The boundary conditions in the vertex are expressed in the $ST$-form as follows,
\begin{equation}\label{ST2}
\begin{pmatrix}
I^{(2)} & T \\
0 & 0
\end{pmatrix}
\begin{pmatrix}
\Psi'_{io} \\
\Psi'_{cd}
\end{pmatrix}
=
\begin{pmatrix}
S & 0 \\
-T^* & I^{(n-2)}
\end{pmatrix}
\begin{pmatrix}
\Psi_{io} \\
\Psi_{cd}
\end{pmatrix}
\,,
\end{equation}
where $S$ is a Hermitian $2\times2$ matrix, $T\in\CC^{2,n-2}$, $I^{(2)},I^{(n-2)}$ are identity matrices of appropriate sizes, $\Psi_{io},\Psi'_{io}$ are the boundary values at input and output line (cf.~\eqref{Psi_io}), and
$$
\Psi_{cd}=\begin{pmatrix}\psi_{r+1}(0) \\ \vdots \\ \psi_n(0) \end{pmatrix}\,; \qquad \Psi'_{cd}=\begin{pmatrix}\psi_{r+1}'(0) \\ \vdots \\ \psi_n'(0) \end{pmatrix}\,.
$$
are the boundary values at controllers and drains.
Relations~\eqref{k_j} imply
\begin{equation}\label{K 2}
\Psi'_{cd}=\im K\Psi_{cd}\,,
\end{equation}
where
$$
K=\diag(k_3,\ldots,k_n)\,.
$$
Identity~\eqref{K 2} allows to eliminate $\Psi_{cd}$ and $\Psi'_{cd}$ from system~\eqref{ST2}. We arrive at dissipative boundary conditions connecting just the input and output,
$$
\Psi'_{io}=(S-\im TKT^*)\Psi_{io}\,.
$$
Formula~\eqref{S} applied on $A_\diss=S-\im TKT^*$ and $B_\diss=I$ leads to the scattering matrix
\begin{equation}\label{Sred2}
\Sc_\diss(E)=-I+\frac{2}{1+\Tr(M(E))+\det(M(E))}\adj(M(E))
\end{equation}
for
\begin{equation}\label{M2}
M(E)=TDT^*-\frac{1}{\im\sqrt{E}}S\,,
\end{equation}
where $D=\diag\left(\frac{k_3}{\sqrt{E}},\cdots,\frac{k_n}{\sqrt{E}}\right)$ for $k_j$ defined in~\eqref{k_j}.
In particular, the transmission amplitude is
\begin{equation}\label{T2}
\Trans(E)=\frac{-2[M(E)]_{21}}{1+\Tr(M(E))+\det(M(E))}\,.
\end{equation}

Once we have derived formula~\eqref{T2}, our next goal is to find requirements on $S$ and $T$ to satisfy conditions~\eqref{const}--\eqref{lim=0}. We may assume without loss of generality that the controlling lines are given numbers~$3,\ldots,q$ for a certain $q\in[3,\ldots,n]$, and edges nos.~$q+1,\ldots,n$ represent drains. We write the matrix $T\in\CC^{2,n-2}$ accordingly in the way
\begin{equation}\label{Tvw}
T=
\begin{pmatrix}
v_1 & w_1 \\
v_2 & w_2
\end{pmatrix}
\end{equation}
with $v_1,v_2\in\CC^{1,q-2}$ and $w_1,w_2\in\CC^{1,n-q}$, where $q-2$ is the number of controllers.
We start from condition~\eqref{lim=0}, i.e., $\lim_{E\to\infty}\Trans(E)=0$. Since $\lim_{E\to\infty}D=I$, equation~\eqref{M2} gives
$$
\lim_{E\to\infty}M(E)=TT^*=
\begin{pmatrix}\|v_1\|^2+\|w_2\|^2 & v_1v_2^*+w_1w_2^* \\ v_2v_1^*+w_2w_1^* & \|v_2\|^2+\|w_2\|^2\end{pmatrix}\,.
$$
Matrix $TT^*$ is Hermitian and positive-definite; thus $\lim_{E\to\infty}\det(M(E))>0$ and $\lim_{E\to\infty}\Tr(M(E))>0$. Consequently, with regard to equation~\eqref{T2}, we have $\lim_{E\to\infty}\Trans(E)=0$ if and only if
\begin{equation}\label{diagonality}
v_2v_1^*+w_2w_1^*=0.
\end{equation}
Now we proceed to condition~\eqref{const}. If the controlling lines support a potential $V$, the matrix $D$ for energies $E\in(0,V)$ equals
$$
D=
\begin{pmatrix}
\im\sqrt{\frac{V}{E}-1}\cdot I^{(q-2)} & 0 \\
0 & I^{(n-q)}
\end{pmatrix}\,.
$$
Formula~\eqref{T2} gives the transmission amplitude for $E\in(0,V)$ in the form
$$
\Trans(E)=\frac{-2\left(w_2w_1^*+\im\sqrt{\frac{V}{E}-1}\cdot v_2v_1^*+\frac{\im}{\sqrt{E}}s_{21}\right)}{a+b\cdot\im\sqrt{\frac{V}{E}-1}+c\cdot\left(\frac{V}{E}-1\right)+d\cdot\sqrt{\frac{V}{E^2}-\frac{1}{E}}+f\cdot\frac{\im}{\sqrt{E}}+\frac{g}{E}}
$$
with
\begin{align*}
a&=1+\|w_1\|^2+\|w_2\|^2+\|w_1\|^2\cdot\|w_2\|^2-|w_2w_1^*|^2\,; \\
b&=\|v_1\|^2+\|v_2\|^2+\|w_1\|^2\cdot\|v_2\|^2+\|v_1\|^2\cdot\|w_2\|^2-v_2v_1^*w_1w_2^*-w_2w_1^*v_1v_2^*\,; \\
c&=-\|v_1\|^2\cdot\|v_2\|^2+|v_2v_1^*|^2\,; \\
d&=-s_{11}\|v_2\|^2-s_{22}\|v_1\|^2+2\Re(s_{21}v_1v_2^*)\,; \\
f&=\Tr(S)+s_{11}\|w_2\|^2+s_{22}\|w_1\|^2-2\Re(s_{21}w_1w_2^*)\,; \\
g&=-\det(S)\,.
\end{align*}
Expressions for $b$ and $\Trans(E)$ can be simplified using equation~\eqref{diagonality},
$$
b=\|v_1\|^2+\|v_2\|^2+\|w_1\|^2\cdot\|v_2\|^2+\|v_1\|^2\cdot\|w_2\|^2+2|v_2v_1^*|^2\,;
$$
\begin{equation}\label{T(E) 2}
\Trans(E)=\frac{2\left(1-\im\sqrt{\frac{V}{E}-1}\right)v_2v_1^*-2\frac{\im}{\sqrt{E}}s_{21}}{a+b\cdot\im\sqrt{\frac{V}{E}-1}+c\cdot\left(\frac{V}{E}-1\right)+d\cdot\sqrt{\frac{V}{E^2}-\frac{1}{E}}+f\cdot\frac{\im}{\sqrt{E}}+\frac{g}{E}}\,.
\end{equation}

\begin{lemma}\label{Lem.nonzero}
Condition~\eqref{const} implies $v_2v_1^*\neq0$.
\end{lemma}

\begin{proof}
We prove the statement by showing that $v_2v_1^*=0$ contradicts~\eqref{const}. Equation $v_2v_1^*=0$ implies $w_2w_1^*=0$ due to equation~\eqref{diagonality}. Therefore, for all $E\in(0,V)$,
$$
\Trans(E)=\frac{-2\frac{\im}{\sqrt{E}}s_{21}}{a+b\cdot\im\sqrt{\frac{V}{E}-1}+c\cdot\left(\frac{V}{E}-1\right)+d\cdot\sqrt{\frac{V}{E^2}-\frac{1}{E}}+f\cdot\frac{\im}{\sqrt{E}}+\frac{g}{E}}\,.
$$
Note that $a=1+\|w_1\|^2+\|w_2\|^2+\|w_1\|^2\cdot\|w_2\|^2\neq0$. Therefore, function $|\Trans(E)|^2$ is either identically zero (for $s_{21}=0$), or non-constant. In both cases condition~\eqref{const} is violated.
\end{proof}
With regard to Lemma~\ref{Lem.nonzero}, we may assume
\begin{equation}\label{nonzero}
v_1\neq0\,,\quad v_2\neq0\,,\quad w_1\neq0\,,\quad w_2\neq0\,.
\end{equation}
We see from the structure of the numerator and the denominator in equation~\eqref{T(E) 2} that satisfying condition~\eqref{const} for all $E<V$ requires
\begin{equation}\label{conditions}
c=0\,,\quad d=0\,,\quad g=0\,,\quad |a|=|b|\,,\quad \frac{f}{b}=\frac{s_{21}}{v_2v_1^*}\,.
\end{equation}
Indeed, when~\eqref{conditions} hold true, we have
$$
\Trans(E)=\frac{2v_2v_1^*}{a}\cdot\frac{1-\im\left(\sqrt{\frac{V}{E}-1}+\frac{s_{21}}{v_2v_1^*}\cdot\frac{1}{\sqrt{E}}\right)}{1+\im\left(\sqrt{\frac{V}{E}-1}+\frac{s_{21}}{v_2v_1^*}\cdot\frac{1}{\sqrt{E}}\right)} \quad\text{for all } E\in(0,V)\,;
$$
hence
\begin{equation}\label{P.2}
\Prob(E)=\left(\frac{2|v_2v_1^*|}{1+\|w_1\|^2+\|w_2\|^2+\|w_1\|^2\cdot\|w_2\|^2-|v_2v_1^*|^2}\right)^2=\const.
\end{equation}
for $E\in(0,V)$.
Now we will examine the system of conditions~\eqref{conditions}. We start from equation $c=0$, which is equivalent to
\begin{equation}\label{lin_dep}
\|v_1\|^2\cdot\|v_2\|^2=|v_2v_1^*|^2\,.
\end{equation}
Due to Cauchy--Schwarz inequality, $v_1,v_2$ are linearly dependent vectors. Furthermore, equation~\eqref{diagonality} together with~\eqref{lin_dep} implies
\begin{equation}\label{|w|}
|w_2w_1^*|=\|v_1\|\cdot\|v_2\|\,.
\end{equation}
Let us proceed to another condition from~\eqref{conditions},  $|a|=|b|$. By virtue of equation~\eqref{|w|}, we can rewrite $|a|=|b|$ in the form
\begin{multline}\label{2nd cond.}
1+\|w_1\|^2+\|w_2\|^2+\|w_1\|^2\cdot\|w_2\|^2-\|v_1\|^2\|v_2\|^2 \\
=\|v_1\|^2+\|v_2\|^2+\|w_1\|^2\cdot\|v_2\|^2+\|v_1\|^2\cdot\|w_2\|^2+2\|v_1\|^2\|v_2\|^2\,,
\end{multline}
which is equivalent to
\begin{equation}\label{simplified}
\left(1+\|w_1\|^2-\|v_1\|^2\right)\left(1+\|w_2\|^2-\|v_2\|^2\right)=4\|v_1\|^2\|v_2\|^2\,.
\end{equation}
We continue to condition $g=0$, which gives
\begin{equation}\label{|s21|}
|s_{21}|=\sqrt{|s_{11}s_{22}|}\,.
\end{equation}
We proceed in~\eqref{conditions} to condition $\frac{f}{b}=\frac{s_{21}}{v_2v_1^*}$. This condition implies in particular that $\frac{s_{21}}{v_2v_1^*}\in\RR$. If we combine this fact with equations~\eqref{|s21|} and~\eqref{lin_dep}, we find
\begin{equation}\label{s21}
s_{21}=\pm\sqrt{|s_{11}s_{22}|}\frac{v_2v_1^*}{\|v_2\|\cdot\|v_1\|}\,.
\end{equation}

Result~\eqref{s21} gives $\frac{s_{21}}{v_2v_1^*}=\pm\frac{\sqrt{|s_{11}s_{22}|}}{\|v_1\|\cdot\|v_2\|}$.
Therefore, condition $\frac{f}{b}=\frac{s_{21}}{v_2v_1^*}$ is equivalent to $\pm f\cdot\|v_1\|\cdot\|v_2|=b\sqrt{|s_{11}s_{22}|}$, i.e.,
\begin{multline}\label{a=b}
\pm\left(\Tr(S)+s_{11}\|w_2\|^2+s_{22}\|w_1\|^2-2\Re(s_{21}w_1w_2^*)\right)\cdot\|v_1\|\cdot\|v_2\| \\
=\left(\|v_1\|^2+\|v_2\|^2+\|w_1\|^2\cdot\|v_2\|^2+\|v_1\|^2\cdot\|w_2\|^2+2|w_2w_1^*|^2\right)\sqrt{|s_{11}s_{22}|}\,.
\end{multline}
We use equations~\eqref{diagonality} and \eqref{s21} to rewrite $s_{21}w_1w_2^*=-s_{21}v_1v_2^*=\mp\sqrt{|s_{11}s_{22}|}\cdot\|v_1\|\cdot\|v_2\|$. Similarly, we rewrite term $|w_2w_1^*|^2$ on the right hand side of~\eqref{a=b} using equation~\eqref{|w|}. As a result of these operations certain terms in equation~\eqref{a=b} cancel, and we get
\begin{multline}\label{modif.}
\pm\left(s_{11}(1+\|w_2\|^2)+s_{22}(1+\|w_1\|^2)\right)\cdot\|v_1\|\cdot\|v_2\| \\
=\left(\|v_1\|^2(1+\|w_2\|^2)+\|v_2\|^2(1+\|w_1\|^2)\right)\cdot\sqrt{|s_{11}s_{22}|}\,.
\end{multline}
The last condition among \eqref{conditions} to be examined is $d=0$. We substitute for $s_{21}$ from equation~\eqref{s21} into the expression for $d$; then $d=0$ is equivalent to
\begin{equation}\label{d=0}
-s_{11}\|v_2\|^2-s_{22}\|v_1\|^2\pm2\sqrt{|s_{11}s_{22}|}\cdot\|v_1\|\cdot\|v_2\|=0\,.
\end{equation}
Now we will examine equation~\eqref{d=0} using equation~\eqref{modif.}. We distinguish three cases.

\paragraph{Case $s_{11}s_{22}=0$.} In this case equation~\eqref{d=0} together with~\eqref{nonzero} implies $s_{11}=s_{22}=0$. Hence $S=0$ due to equation~\eqref{|s21|}. Consequently, condition~\eqref{modif.} is always satisfied for $s_{11}s_{22}=0$.

\paragraph{Case $s_{11}s_{22}>0$.} Equation~\eqref{d=0} is equivalent to
$$
\left(\sqrt{|s_{11}|}\|v_2\|\mp\sgn(s_{11})\sqrt{|s_{22}|}\|v_1\|\right)^2=0\,,
$$
hence, due to~\eqref{nonzero},
\begin{equation}\label{pozit}
\pm\sgn(s_{11})=1 \quad\text{and}\quad \sqrt{\frac{s_{11}}{s_{22}}}=\frac{\|v_1\|}{\|v_2\|}\,.
\end{equation}
Similarly, equation~\eqref{modif.} is equivalent to
\begin{multline}\label{fb pozit}
\pm\sgn(s_{11})\left(|s_{11}|(1+\|w_2\|^2)+|s_{22}|(1+\|w_1\|^2)\right)\|v_1\|\cdot\|v_2\| \\
=\left(\|v_1\|^2(1+\|w_2\|^2)+\|v_2\|^2(1+\|w_1\|^2)\right)\cdot\sqrt{|s_{11}s_{22}|}\,.
\end{multline}
When we substitute relations~\eqref{pozit} into equation~\eqref{fb pozit}, we get an identity. In other words, condition~\eqref{modif.} is always satisfied for $s_{11}s_{22}>0$.

\paragraph{Case $s_{11}s_{22}<0$.} Equation~\eqref{d=0} is equivalent to
\begin{multline*}
\left(\sqrt{|s_{11}|}\|v_2\|+(\sqrt{2}\mp\sgn(s_{11}))\sqrt{|s_{22}|}\|v_1\|\right) \\
\cdot\left(\sqrt{|s_{11}|}\|v_2\|-(\sqrt{2}\pm\sgn(s_{11}))\sqrt{|s_{22}|}\|v_1\|\right)=0
\end{multline*}
hence
\begin{equation}\label{negat}
\sqrt{\frac{|s_{11}|}{|s_{22}|}}=\left(\sqrt{2}\pm\sgn(s_{11})\right)\frac{\|v_1\|}{\|v_2\|}\,.
\end{equation}
Equation~\eqref{modif.} is equivalent to
\begin{multline}\label{fb negat}
\pm\sgn(s_{11})\left(|s_{11}|(1+\|w_2\|^2)-|s_{22}|(1+\|w_1\|^2)\right)\|v_1\|\|v_2\| \\
=\left(\|v_1\|^2(1+\|w_2\|^2)+\|v_2\|^2(1+\|w_1\|^2)\right)\cdot\sqrt{|s_{11}s_{22}|}\,.
\end{multline}
When we use relation~\eqref{negat} in equation~\eqref{fb negat}, we get the equation
$$
\pm\sgn(s_{11})\sqrt{2}\left(\|v_1\|^2(1+\|w_2\|^2)-\|v_2\|^2(1+\|w_1\|^2)\right)=0\,.
$$
Hence
\begin{equation}\label{equiv12}
\frac{\|v_1\|^2}{1+\|w_1\|^2}=\frac{\|v_2\|^2}{1+\|w_2\|^2}\,.
\end{equation}
We apply equivalence~\eqref{equiv12} in equation~\eqref{simplified} and get
$$
\left(1+\|w_1\|^2-\|v_1\|^2\right)^2=4\|v_1\|^4\,;
$$
hence
\begin{equation}\label{vw negat}
\|w_1\|^2=3\|v_1\|^2-1\,,\qquad \|w_2\|^2=3\|v_2\|^2-1\,.
\end{equation}

At this stage we have finished the study of conditions~\eqref{lim=0} and \eqref{const}. It remains to check condition~\eqref{drop}. It is straightforward to derive the formula
$$
\Prob(E)=\left(\frac{2|v_2v_1^*|}{(1+\|w_1\|^2)(1+\|w_2\|^2)-|w_2w_1^*|^2}\right)^2\cdot\frac{\left(1-\sqrt{1-\frac{V}{E}}\right)^2+\frac{|s_{21}|^2}{\|v_1\|^2\|v_2\|^2}\cdot\frac{1}{E}}{\left(1+\sqrt{1-\frac{V}{E}}\right)^2+\frac{|s_{21}|^2}{\|v_1\|^2\|v_2\|^2}\cdot\frac{1}{E}}
$$
for all $E>V$. It is easy to verify that $\lim_{E\searrow V}\Prob'(E)=-\infty$.

Let us summarize the results in the following theorem.
\begin{theorem}\label{Thm.flat}
Consider a star graph with a vertex coupling described by boundary conditions~\eqref{ST2}. The transmission probability in the input-output channel satisfies conditions~\eqref{const}--\eqref{lim=0} if and only if vectors $v_1$ and $v_2$ in matrix $T$ \eqref{Tvw} are linearly dependent, vectors $w_1,w_2$ obey requirements~\eqref{diagonality} and~\eqref{simplified}, and one of the following three cases holds true.
\begin{itemize}
\item $S=0$;
\item $s_{11}s_{22}>0$ and
$$
\sqrt{\frac{s_{11}}{s_{22}}}=\frac{\|v_1\|}{\|v_2\|}\,,\qquad s_{21}=\sgn(s_{11})\cdot\sqrt{s_{11}s_{22}}\frac{v_2v_1^*}{\|v_1\|\cdot\|v_2\|}\,.
$$
\item $s_{11}s_{22}<0$, condition~\eqref{vw negat} is satisfied, and matrix $S$ obeys conditions
$$
\sqrt{\frac{|s_{11}|}{|s_{22}|}}=\left(\sqrt{2}\pm\sgn(s_{11})\cdot1\right)\frac{\|v_1\|}{\|v_2\|} \quad\text{and}\quad s_{21}=\pm\sqrt{|s_{11}s_{22}|}\frac{v_2v_1^*}{\|v_1\|\cdot\|v_2\|}\,.
$$
\end{itemize}
\end{theorem}
An example of transmission probability function for $S,T$ chosen according to Theorem~\ref{Thm.flat} is shown in Figure~\ref{Fig.0.16}.
\begin{figure}[h]
\begin{center}
\includegraphics[angle=0,width=0.45\textwidth]{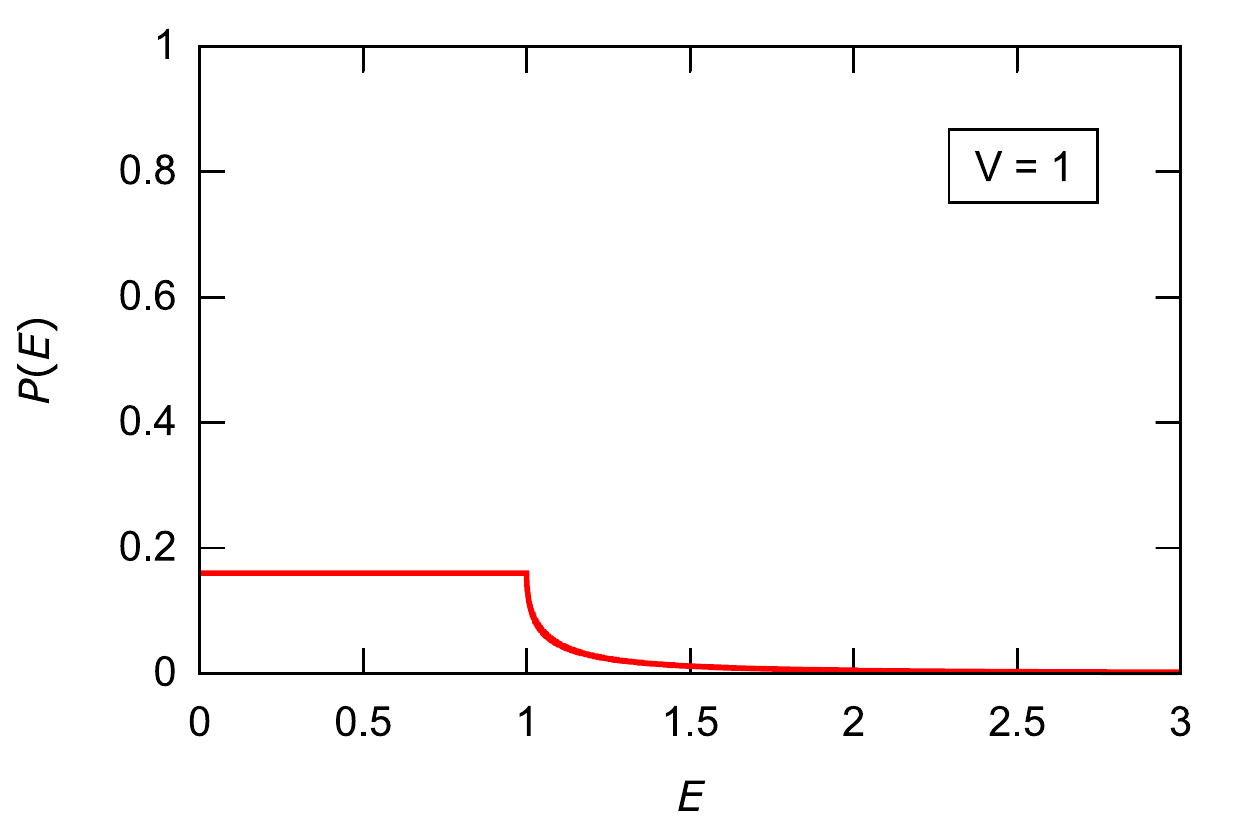}
\caption{An example of a transmission probability featuring a flat passband. The function is obtained for the choice $\|v_1\|^2=\frac{5}{8}$, $v_2=v_1$, $\|w_1\|^2=\|w_2\|^2=\frac{7}{8}$, $|w_2w_1^*|^2=\frac{5}{8}$, $S=0$ and for the controlling potential $V=1$.}\label{Fig.0.16}
\end{center}
\end{figure}

Theorem~\ref{Thm.flat} describes the structure of matrices $S,T$ in boundary conditions~\eqref{ST2} for which the star graph works as a band-pass filter with a flat passband. In the rest of the section we will find the maximal possible value of $|\Trans(E)|$ in the ``flat band'' interval $(0,V)$, and we will characterize the corresponding matrices $S,T$. With regard to equation~\eqref{P.2}, we have to find the maximum of the quantity
\begin{equation}\label{f.max}
\left(\frac{2|v_2v_1^*|}{1+\|w_1\|^2+\|w_2\|^2+\|w_1\|^2\cdot\|w_2\|^2-|w_2w_1^*|^2}\right)^2
\end{equation}
under conditions given in Theorem~\ref{Thm.flat}. Note that the expression~\eqref{f.max} is independent of $S$, and the entries of $S$ can be calculated after $T$ is fixed. Therefore, we will at first find the maximum of expression~\eqref{f.max} under conditions~\eqref{diagonality} and~\eqref{simplified}, whereas matrix $S$ will be calculated later.
We denote
$$
\|w_1\|^2=x\,,\; \|w_2\|^2=y\,,\; \|v_1\|^2=z
$$
and
$$
|v_2v_1^*|^2=xyu \quad \text{for a certain $u\in(0,1]$},
$$
which is possible due to $|v_2v_1^*|^2=|w_2w_1^*|^2\leq\|w_1\|^2\cdot\|w_2\|^2$.
We express $\|v_2\|^2$ using equations~\eqref{diagonality} and \eqref{lin_dep},
$$
\|v_2\|^2=\frac{\|v_1\|^2\cdot\|v_2\|^2}{\|v_1\|^2}=\frac{|v_2v_1^*|^2}{\|v_1\|^2}=\frac{|w_2w_1^*|^2}{\|v_1\|^2}=\frac{xyu}{z}\,.
$$
We shall find the maximum of the function
$$
F(x,y,z,u)=\left(\frac{2\sqrt{xyu}}{1+x+y+xy-xyu}\right)^2
$$
(cf.~\eqref{f.max}) under condition
$$
\left(1+x-z\right)\left(1+y-\frac{xyu}{z}\right)=4xyu\,.
$$
We proceed in a standard way. We introduce the Langrage function
$$
\mathcal{L}(x,y,z,u,\lambda)=\frac{2\sqrt{xyu}}{1+x+y+xy-xyu}-\lambda\cdot\left[\left(1+x-z\right)\left(1+y-\frac{xyu}{z}\right)-4xyu\right]
$$
and solve the system
\begin{equation}\label{L system}
\frac{\partial\mathcal{L}}{\partial x}=0\,,\quad\frac{\partial\mathcal{L}}{\partial y}=0\,,\quad\frac{\partial\mathcal{L}}{\partial z}=0\,,\quad\frac{\partial\mathcal{L}}{\partial u}=0\,.
\end{equation}
It turns out that \eqref{L system} has no solution. Therefore, we shall search for the maximum of $F$ at the boundary of its domain, i.e., for $u=1$. If we fix $u=1$ and solve the system $\frac{\partial\mathcal{L}}{\partial x}=\frac{\partial\mathcal{L}}{\partial y}=\frac{\partial\mathcal{L}}{\partial z}=0$, we obtain
$$
x=y=z=\frac{1}{2}\,.
$$
Hence we find the sought maximum of the function $F$,
$$
F\left(\frac{1}{2},\frac{1}{2},\frac{1}{2},1\right)=\frac{1}{4}\,.
$$
Note that $u=1$ implies $|w_2w_1^*|=\|w_1\|\cdot\|w_2\|$, i.e., $w_1,w_2$ are linearly dependent. 

\begin{theorem}\label{Thm.max}
The maximal transmission probability of a band-pass filter with flat passband, constructed upon a vertex with boundary conditions~\eqref{ST2}, is $\frac{1}{4}$. It is obtained for
\begin{equation}\label{Tmax}
T=
\begin{pmatrix}
v & w \\
\alpha v & -\alpha w
\end{pmatrix}
\qquad\text{for } \|v\|=\|w\|=\frac{1}{\sqrt{2}}\,,\quad |\alpha|=1\,,
\end{equation}
and
\begin{equation}\label{Smax}
S=s\begin{pmatrix}1&\bar{\alpha}\\ \alpha&1\end{pmatrix} \qquad\text{or}\qquad
S=s\begin{pmatrix}1\pm\sqrt{2}&\bar{\alpha}\\ \alpha&1\mp\sqrt{2}\end{pmatrix}
\qquad\text{for } s\in\RR\,.
\end{equation}
\end{theorem}

\begin{proof}
According to calculations above, the maximal transmission probability is $\frac{1}{4}$, and this value is attained for $\|v_1\|=\|v_2\|=\|w_1\|=\|w_2\|=\frac{1}{\sqrt{2}}$. Vectors $v_1,v_2$ are linearly dependent due to Theorem~\ref{Thm.flat}; hence $v_2=\alpha\cdot v_1$ for an $\alpha$ satisfying $|\alpha|=1$. Equation~\eqref{diagonality} implies $w_2=-\alpha\cdot w_1$. Furthermore, equations listed in Theorem~\ref{Thm.flat} imply that either $S=0$, or the entries of $S$ satisfy
$$
s_{11}=s_{22}=s\,,\quad s_{21}=s\cdot\alpha
$$
for a certain $s\neq0$, or
\begin{align*}
s_{11}&=s\cdot\left(\sqrt{2}\pm\sgn(s)\cdot1\right) \\
s_{22}&=-s\cdot\left(\sqrt{2}\mp\sgn(s)\cdot1\right) \\
s_{21}&=\pm|s|\cdot\alpha
\end{align*}
for a certain $s\neq0$.
It is easy to check that all the cases above are fully covered by formulas \eqref{Smax}.
\end{proof}
Figure~\ref{Fig.maximal} shows two examples of the transmission probability functions obtained for $S,T$ obeying conditions from Theorem~\ref{Thm.max}.
\begin{figure}[h]
\begin{center}
\includegraphics[angle=0,width=0.45\textwidth]{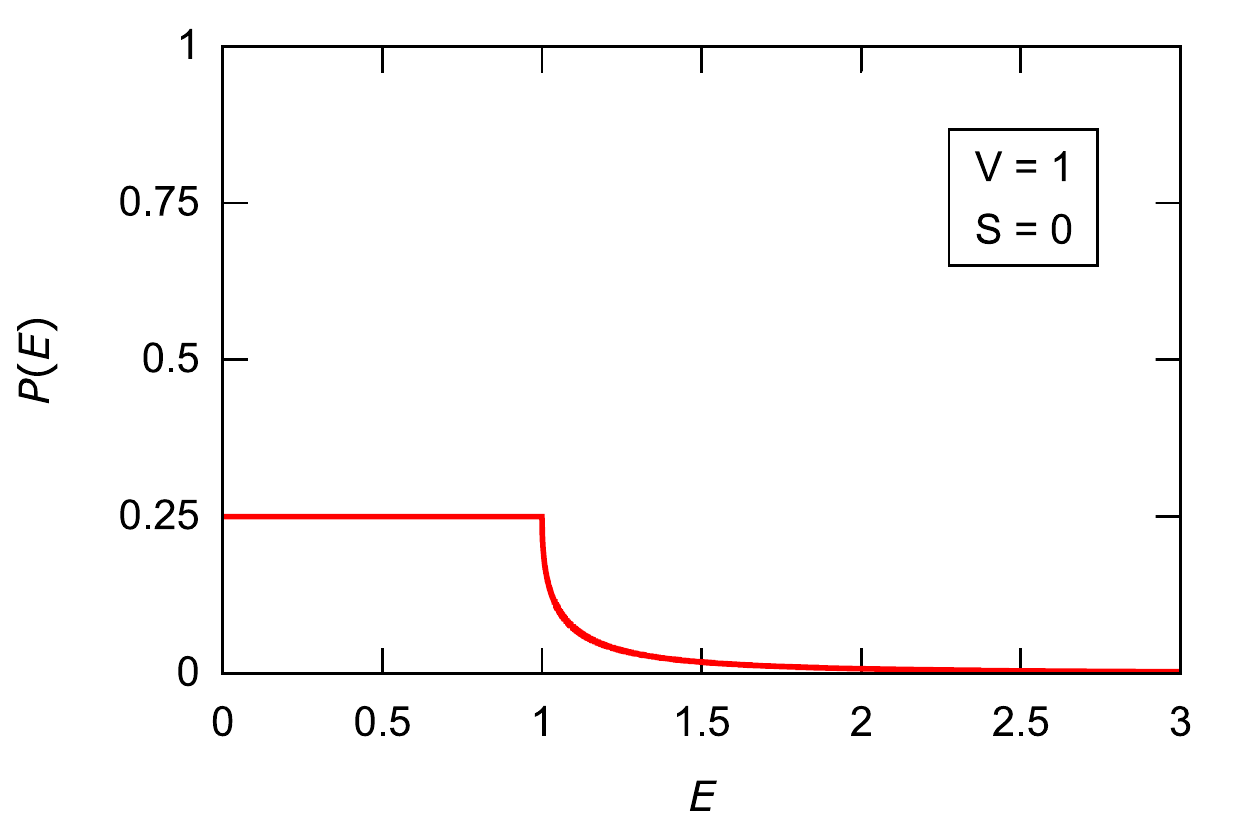}
\quad
\includegraphics[angle=0,width=0.45\textwidth]{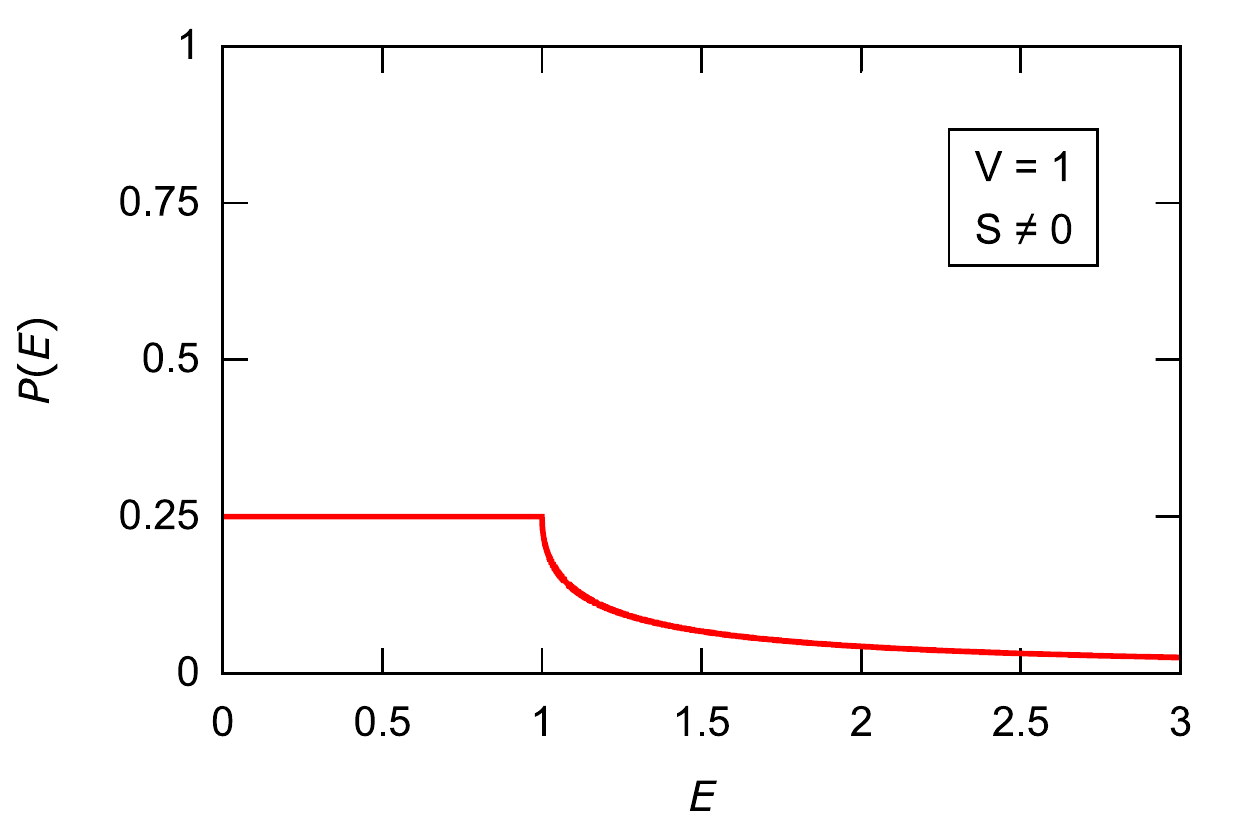}
\caption{The maximal transmission probability in the passband for the controlling potential $V=1$. The graphs display the function $\Prob(E)$ for $T$ given by equation~\eqref{Tmax} and $S=\left(\protect\begin{smallmatrix}
0 & 0 \protect\\
0 & 0
\protect\end{smallmatrix}\right)$ (left) and $S=\frac{1}{2}\left(\protect\begin{smallmatrix}
1 & 1 \protect\\
1 & 1
\protect\end{smallmatrix}\right)$ (right).}\label{Fig.maximal}
\end{center}
\end{figure}

\begin{remark}
Matrix $T$ given by equation~\eqref{Tmax} together with $S=0$ generalizes an earlier result. In~\cite{TC12}, a graph consisting of the input, output, one controlling edge and one drain, coupled in a vertex by scale invariant boundary conditions with
$$
T=
\frac{1}{\sqrt{2}}\begin{pmatrix}
1 & 1 \\
1 & -1
\end{pmatrix}\,,
$$
was examined. It was demonstrated that the transmission probability is constant in the interval $(0,V)$ and quickly decreases towards zero as $E$ exceeds the controlling potential $V$. Theorem~\ref{Thm.max} implies that the flat-band behaviour persists even if the scale invariance of the coupling is broken. This is a quite surprising fact.
\end{remark}

\section{Case $r\geq3$}\label{Sect.3}

The ideas demonstrated in previous sections can be used for treating vertex couplings with $r=\rank(B)\geq3$ as well.
Section~\ref{Sect.LD} implies that if the sought band-pass filter with flat passband exists, then the columns of $B$ that correspond to the input and output line need to be linearly independent.
This allows us to express the boundary conditions in the vertex in the $ST$-form as follows,
\begin{equation}\label{FTdec3}
\begin{pmatrix}
I^{(2)} & 0 & T_1 \\
0 & I^{(r-2)} & T_2 \\
0 & 0 & 0
\end{pmatrix}
\begin{pmatrix}
\Psi'_{io} \\
\Psi'_{cd}
\end{pmatrix}
=
\begin{pmatrix}
0 & 0 & 0 \\
0 & 0 & 0 \\
-T_1^* & -T_2^* & I^{(n-r)}
\end{pmatrix}
\begin{pmatrix}
\Psi_{io} \\
\Psi_{cd}
\end{pmatrix}
\,,
\end{equation}
where $\Psi_{io},\Psi'_{io}$ are the boundary values at input and output (cf.~\eqref{Psi_io}), $\Psi_{cd},\Psi'_{cd}$ are the boundary values at controlling lines and drains (controlling edges and drains not being distinguished now), $I^{(2)},I^{(r-2)},I^{(n-r)}$ are identity matrices of given orders, and $T=\begin{pmatrix}T_1\\T_2\end{pmatrix}\in\CC^{r,n-r}$.

Relation~\eqref{proportional} implies
\begin{equation}\label{prop3}
\Psi'_{cd}=\im\begin{pmatrix}K_2&0\\0&K_3\end{pmatrix}\Psi_{cd}\,,
\end{equation}
where $K_2=\diag(k_3,\ldots,k_r)$ and $K_3=\diag(k_{r+1},\ldots,k_n)$.
Elimination of $\Psi_{cd}$ and $\Psi'_{cd}$ from system~\eqref{FTdec3} using equation~\eqref{prop3} leads to the conditions
$$
\Psi'_{io}=-\im T_1\left(K_3+T_2^*K_2^{-1}T_2\right)^{-1}T_1^*\Psi_{io}\,.
$$
Formula~\eqref{S} gives the dissipative scattering matrix
\begin{equation}\label{Sred3}
\Sc_\diss(E)=-I+\frac{2}{1+\Tr(M(E))+\det(M(E))}\adj(M(E))\,,
\end{equation}
where the matrix $M(E)$ is given as
$$
M(E)=T_1\left(D_3+T_2^*D_2^{-1}T_2\right)^{-1}T_1^*
$$
with $D_2=\diag\left(\frac{k_3}{\sqrt{E}},\cdots,\frac{k_r}{\sqrt{E}}\right)$ and $D_3=\diag\left(\frac{k_{r+1}}{\sqrt{E}},\cdots,\frac{k_n}{\sqrt{E}}\right)$.
Consequently, the transmission amplitude is
\begin{equation}\label{T3}
\Trans(E)=\frac{-2[M(E)]_{21}}{1+\Tr(M(E))+\det(M(E))}\,.
\end{equation}
We require $\lim_{E\to\infty}\Trans(E)=0$ according to~\eqref{lim=0}. Since we have $\lim_{E\to\infty}D_2=I$ and $\lim_{E\to\infty}D_3=I$, we get
$$
\lim_{E\to\infty}M(E)=T_1\left(I+T_2^*T_2\right)^{-1}T_1^*\,.
$$
The matrix on the right hand side is Hermitian and positive-definite. The denominator of~\eqref{Sred3} thus tends to a positive number greater than $1$ as $E\to\infty$. Therefore, equation~\eqref{T3} gives the equivalence
$$
\lim_{E\to\infty}\Trans(E)=0 \quad\Leftrightarrow\quad T_1\left(I+T_2^*T_2\right)^{-1}T_1^* \text{ is diagonal}.
$$
To sum up, a quantum star graph with the vertex coupling given by boundary conditions~\eqref{FTdec3} can work as a band-pass filter with flat passband only if $T_1\left(I+T_2^*T_2\right)^{-1}T_1^*$ is a diagonal matrix.

Analyzing condition~\eqref{const} needs to distinguish controllers and drains in both sets $\{3,\ldots,r\}$ and $\{r+1,\ldots,n\}$, which would make the problem more intricate. Therefore, the case $r\geq3$ in general will not be addressed in this paper; nevertheless, the method presented in Sections~\ref{Sect.1}--\ref{Sect.2} is in principle applicable.

\begin{remark}
Although we focused on graphs working as spectral band-pass filters with flat passbands, the same approach can be used more generally. Taking advantage of the $ST$-form of boundary conditions, one can explore and design quantum graphs with various other special transmission characteristics, such as filters having a sharp peak in $\Prob(E)$ at a certain energy.
\end{remark}




\end{document}